\documentclass{llncs}

\usepackage{amssymb,amsmath,graphicx}
\usepackage{color,enumerate,comment,boxedminipage,tkz-graph}
\usepackage{pgf,float,hyperref}
\usepackage{subcaption}
\usepackage[T1]{fontenc}

\spnewtheorem{remark2}{Remark}{\bfseries}{}

\newcounter{ctrclaim}[theorem]
\newcounter{ctrcase}[theorem]
\newcounter{ctrstep}[theorem]
\newcounter{ctrsubstep}[ctrstep]
\newcounter{ctrsubsubstep}[ctrsubstep]

\newcommand\displaycase[1]{{\bfseries#1}}
\newcommand\faketheorem[1]{{\bfseries#1}}
\newcommand{\clm}[1]{\medskip\phantomsection\refstepcounter{ctrclaim}\noindent\displaycase{Claim \thectrclaim. }{\itshape #1}}
\newcommand{\thmcase}[1]{\medskip\phantomsection\refstepcounter{ctrcase}\noindent\displaycase{Case \thectrcase. }{\itshape #1}\\}
\newcommand{\thmstep}[1]{\medskip\phantomsection\refstepcounter{ctrstep}\noindent\displaycase{Step \thectrstep. }{\itshape #1}\\}
\newcommand{\thmsubstep}[1]{\medskip\phantomsection\refstepcounter{ctrsubstep}\noindent\displaycase{Step \thectrstep.\thectrsubstep. }{\itshape #1}\\}
\newcommand{\thmsubsubstep}[1]{\medskip\phantomsection\refstepcounter{ctrsubsubstep}\noindent\displaycase{Step \thectrstep.\thectrsubstep\alph{ctrsubsubstep}. }{\itshape #1}\\}

\newcommand{\NP}{{\sf NP}}

\newcommand{\problemdef}[3]{
	\begin{center}
		\begin{boxedminipage}{.99\textwidth}
			\textsc{{#1}}\\[2pt]
			\begin{tabular}{ r p{0.8\textwidth}}
				\textit{~~~~Instance:} & {#2}\\
				\textit{Question:} & {#3}
			\end{tabular}
		\end{boxedminipage}
	\end{center}
}
\title{On Cycle Transversals and Their Connected Variants in the Absence of a Small Linear Forest\thanks{The research in this paper received support from the Leverhulme Trust (RPG-2016-258).
The last author was supported by Polish National Science Centre grant no. 2018/31/D/ST6/00062.
An extended abstract of this paper appeared in the proceedings of FCT 2019~\cite{FJPP19}.}
}

\author{Konrad K. Dabrowski\inst{1}
\and
Carl Feghali\inst{2}
\and
Matthew Johnson\inst{1}
\and\\
Giacomo Paesani\inst{1}
\and
Dani\"el Paulusma\inst{1}
\and
Pawe{\l} Rz{\k{a}}\.{z}ewski\inst{3}}

\institute{Department of Computer Science, Durham University, UK \\ \email{\{konrad.dabrowski,matthew.johnson2,giacomo.paesani,daniel.paulusma\}@durham.ac.uk} \and Department of Informatics, University of Bergen, Norway \email{carl.feghali@uib.no} \and Faculty of Mathematics and Information Science, Warsaw University of Technology, Warsaw, Poland \email{p.rzazewski@mini.pw.edu.pl}}

\begin{document}

\maketitle

\begin{abstract}
A graph is $H$-free if it contains no induced subgraph isomorphic to~$H$.
We prove new complexity results for the two classical cycle transversal problems {\sc Feedback Vertex Set} and {\sc Odd Cycle Transversal} by showing that they can be solved in polynomial time on $(sP_1+\nobreak P_3)$-free graphs for every integer $s\geq 1$.
We show the same result for the variants {\sc Connected Feedback Vertex Set} and {\sc Connected Odd Cycle Transversal}.
We also prove that the latter two problems are polynomial-time solvable on cographs; this was already known for {\sc Feedback Vertex Set} and {\sc Odd Cycle Transversal}.
We complement these results by proving that {\sc Odd Cycle Transversal} and {\sc Connected Odd Cycle Transversal} are \NP-complete on $(P_2+\nobreak P_5,P_6)$-free graphs.
\end{abstract}

\section{Introduction}\label{s-intro}

Graph transversal problems play a central role in Theoretical Computer Science.
To define the notion of a graph transversal, let~${\cal H}$ be a family of graphs, $G=(V,E)$ be a graph and $S\subseteq V$ be a subset of vertices of~$G$.
The graph $G-S$ is obtained from~$G$ by removing all vertices of~$S$ and all edges incident to vertices in~$S$.
We say that~$S$ is an \emph{${\cal H}$-transversal} of~$G$ if $G-S$ is \emph{${\cal H}$-free}, that is, if $G-S$ contains no induced subgraph isomorphic to a graph of~${\cal H}$.
In other words, $S$ intersects every induced copy of every graph of~${\cal H}$ in~$G$.
Let~$C_r$ and~$P_r$ denote the cycle and path on~$r$ vertices, respectively.
Then~$S$ is a \emph{vertex cover}, \emph{feedback vertex set}, or \emph{odd cycle transversal} if~$S$ is an ${\cal H}$-transversal for, respectively, ${\cal H}=\{P_2\}$ (that is, $G-S$ is edgeless), ${\cal H}=\{C_3,C_4,\ldots\}$ (that is, $G-S$ is a forest), or ${\cal H}=\{C_3,C_5,\ldots\}$ (that is, $G-S$ is bipartite).

Usually the goal is to find a transversal of minimum size in some given graph.
In this paper we focus on the decision problems corresponding to the three transversals defined above.
These are the {\sc Vertex Cover}, {\sc Feedback Vertex Set} and {\sc Odd Cycle Transversal} problems, which are to decide whether a given graph has a vertex cover, feedback vertex set or odd cycle transversal, respectively, of size at most~$k$ for some given positive integer~$k$.
Each of these three problems is well studied and is well known to be \NP-complete.

We may add further constraints to a transversal.
In particular, we may require a transversal of a graph~$G$ to be \emph{connected}, that is, to induce a connected subgraph of~$G$.
The corresponding decision problems for the three above transversals are then called {\sc Connected Vertex Cover}, {\sc Connected Feedback Vertex Set} and {\sc Connected Odd Cycle Transversal}, respectively.

Garey and Johnson~\cite{GJ77} proved that {\sc Connected Vertex Cover} is \NP-complete even on planar graphs of maximum degree~$4$ (see, for example,~\cite{FM09,PH08,WKO91} for \NP-completeness results for other graph classes).
Grigoriev and Sitters~\cite{GS09} proved that {\sc Connected Feedback Vertex Set} is \NP-complete even on planar graphs with maximum degree~$9$. More recently,
Chiarelli et al.~\cite{CHJMP18} proved that {\sc Connected Odd Cycle Transversal} is \NP-complete even on graphs of arbitrarily large girth and on line graphs.

As all three decision problems and their connected variants are \NP-complete, 
we can consider how to restrict the input to some special graph class in order to achieve tractability.
Note that this approach is in line with the aforementioned results in the literature, where \NP-completeness was proven on special graph classes. It is also in line with with, for instance, polynomial-time results for {\sc Connected Vertex Cover} by Escoffier, Gourv\`es and Monnot~\cite{EGM10} (for chordal graphs) and Ueno, Kajitani and Gotoh~\cite{UKG88} (for graphs of maximum degree at most~$3$ and trees).

Just as in most of these papers, we consider \emph{hereditary} graph classes, that is, graph classes closed under vertex deletion.
Hereditary graph classes form a rich framework that captures many well-studied graph classes.
It is not difficult to see that every hereditary graph class~${\cal G}$ can be characterized by a (possibly infinite) set~${\cal F}_{\cal G}$ of forbidden induced subgraphs.
If $|{\cal F}_{\cal G}|=1$, say ${\cal F}=\{H\}$, then~${\cal G}$ is said to be \emph{monogenic}, and every graph $G\in {\cal G}$ is said to be \emph{$H$-free}.
Considering monogenic graph classes can be seen as a natural first step for increasing our knowledge of the complexity of an \NP-complete problem in a \emph{systematic} way. Hence, we consider the following research question:

\medskip
\noindent
{\it How does the structure of a graph~$H$ influence the computational complexity of a graph transversal problem for input graphs that are $H$-free?}

\medskip
\noindent
Note that different graph transversal problems may behave differently on some class of $H$-free graphs.
However, the general strategy for obtaining complexity results is to first try to prove that the restriction to $H$-free graphs is \NP-complete whenever~$H$ contains a cycle or the claw (the 4-vertex star).
This is usually done by showing, respectively, that the problem is \NP-complete on graphs of arbitrarily large {\it girth} (length of a shortest cycle) and on line graphs, which form a subclass of claw-free graphs.
If this is the case, then we are left to consider the case when~$H$ does not contain a cycle, implying that~$H$ is a forest, and does not contain a claw either, implying that~$H$ is a \emph{linear forest}, that is, the disjoint union of one or more paths.

\subsection{The Graph $H$ Contains a Cycle or Claw}\label{s-claw}

It follows from Poljak's construction~\cite{Po74} that {\sc Vertex Cover} is \NP-complete on graphs of arbitrarily large girth. Hence, {\sc Vertex Cover} is \NP-complete on $H$-free graphs if~$H$ contains a cycle.
However, {\sc Vertex Cover} becomes polynomial-time solvable when restricted to claw-free graphs~\cite{Mi80,Sh80}. 
In contrast, the other five problems {\sc Connected Vertex Cover}, {\sc (Connected) Feedback Vertex Set} and {\sc (Connected) Odd Cycle Transversal} are all \NP-complete on graphs of arbitrarily large girth and on line graphs; see Table~\ref{t-thetable}.
Hence, for these five problems, it remains to consider only the case when~$H$ is a linear forest.

\subsection{The Graph $H$ Is a Linear Forest}

In this paper, we focus on proving new complexity results for {\sc Feedback Vertex Set}, {\sc Connected Feedback Vertex Set}, {\sc Odd Cycle Transversal} and {\sc Connected Odd Cycle Transversal} on $H$-free graphs.
It follows from Section~\ref{s-claw} that we may assume that~$H$ is a linear forest.
Below we first discuss the known polynomial-time solvable cases.
As we will use algorithms for {\sc Vertex Cover} and {\sc Connected Vertex Cover} as subroutines for our new algorithms, we include these two problems in our discussion.

For every $s\geq 1$, {\sc Vertex Cover} (by combining the results of~\cite{BY89,TIAS77}) and {\sc Connected Vertex Cover}~\cite{CHJMP18} are polynomial-time solvable on $sP_2$-free graphs.\footnote{The graph $G+\nobreak H$ is the disjoint union of graphs~$G$ and~$H$ and~$sG$ is the disjoint union of~$s$ copies of~$G$; see Section~\ref{s-pre}.}
Moreover, {\sc Vertex Cover} is also polynomial-time solvable on $(sP_1+\nobreak P_6)$-free graphs, for every $s\geq 0$~\cite{GKPP19}, as is the case for {\sc Connected Vertex Cover} on $(sP_1+\nobreak P_5)$-free graphs~\cite{JPP18}.
Their complexity on $P_r$-free graphs is unknown for $r\geq 7$ and $r\geq 6$, respectively.

Both {\sc Feedback Vertex Set} and {\sc Odd Cycle Transversal} are polynomial-time solvable on permutation graphs~\cite{BK85}, and thus on $P_4$-free graphs.
Recently, Okrasa and Rz{\k{a}}\.{z}ewski~\cite{OR19} proved that {\sc Odd Cycle Transversal} is \NP-complete on $P_{13}$-free graphs.
A small modification of their construction yields the same result for {\sc Connected Odd Cycle Transversal}.
The complexity of {\sc Feedback Vertex Set} and {\sc Connected Feedback Vertex Set} is unknown when restricted to $P_r$-free graphs for $r\geq 5$.
For every $s\geq 1$, both problems and their connected variants are polynomial-time solvable on $sP_2$-free graphs~\cite{CHJMP18}, using the price of connectivity for feedback vertex set~\cite{BHKP17,HJMP16}.\footnote{The price of connectivity concept was introduced by Cardinal and Levy~\cite{CL10} for vertex cover; see also, for example,~\cite{Ca19,CCFS14,CS14}.}

\subsection{Our Results}\label{s-our}

In Section~\ref{s-p4} we prove that {\sc Connected Feedback Vertex Set} and {\sc Connected Odd Cycle Transversal} are polynomial-time solvable on $P_4$-free graphs, just as is the case for {\sc Feedback Vertex Set} and {\sc Odd Cycle Transversal}.
In Section~\ref{s-p1p3} we prove that for every $s\geq 1$, these four problems are all polynomial-time solvable on $(sP_1+\nobreak P_3)$-free graphs; see also Table~\ref{t-thetable}.
Finally, in Section~\ref{s-hard}, we show that {\sc Odd Cycle Transversal} and {\sc Connected Odd Cycle Transversal} are \NP-complete on $(P_2+\nobreak P_5,P_6)$-free graphs, that is, graphs that are both $(P_2+\nobreak P_5)$-free and $P_6$-free.

\begin{table}[h]
\centering
\begin{tabular}{|l|l|l|l|l|l|l|}
\hline
& girth~$p$ & line graphs & $sP_2$-free &$P_4$-free &$sP_1+P_r$-free\\[-1pt]
\hline
{\sc Vertex Cover} & \NP-c~\cite{Po74}\; &P\hspace*{4.7mm} \cite{Mi80,Sh80}&P~\cite{BY89,TIAS77}&P&P: $s\geq 0$, $r=6$~\cite{GKPP19}\\[-1pt]
\hline
{\sc Feedback Vertex Set} & \NP-c~\cite{Po74}\;  &\NP-c~\cite{Sp83} &P~\cite{CHJMP18}&P~\cite{BK85}&P: $s\geq 0$, $r=3$$*$\\[-1pt]	
\hline
{\sc Odd Cycle Transversal} &\NP-c~\cite{CHJMP18} &\NP-c~\cite{CHJMP18} &P~\cite{CHJMP18}&P~\cite{BK85}&P: $s\geq 0$, $r=3$$*$\\[-1pt]	
\hline
{\sc Con. Vertex Cover}  &\NP-c~\cite{Mu17} &\NP-c~\cite{Mu17} &P~\cite{CHJMP18}&P&P: $s\geq 0$, $r=5$~\cite{JPP18}\\[-1pt]
\hline
{\sc Con. Feedback Vertex Set} &\NP-c~\cite{CHJMP18}&\NP-c~\cite{CHJMP18}&P~\cite{CHJMP18} &P$*$&P: $s\geq 0$, $r=3$$*$\\[-1pt]
\hline
{\sc  Con. Odd Cycle Transversal} &\NP-c~\cite{CHJMP18} &\NP-c~\cite{CHJMP18}&P~\cite{CHJMP18} &P$*$&P: $s\geq 0$, $r=3$$*$\\[-1pt]
\hline			
\end{tabular}
\vspace*{2mm}
\caption{The complexities of the three connected transversal problems together with the original transversal problems on graphs of girth at least~$p$ for every (fixed) constant $p\geq\nobreak 3$, on line graphs,
and on $H$-free graphs for various linear forests~$H$.
In particular, {\sc Feedback Vertex Set} can be shown to be \NP-complete on graphs of arbitrarily large girth by using Poljak's construction (see~\cite{BDFJP19,MPRS12}). We also note that  Munro~\cite{Mu17} showed that {\sc Feedback Vertex Set} is 
\NP-complete even on line graphs of planar cubic bipartite graphs. Unreferenced results directly follow from other results in the table, and results marked with $*$ are new results proven in this paper. Our two other new results, namely that {\sc Odd Cycle Transversal} and {\sc Connected Odd Cycle Transversal} are \NP-complete on $(P_2+\nobreak P_5,P_6)$-free graphs, are {\bf not} included in the table.}\label{t-thetable}
\end{table}

\noindent
To prove our polynomial-time results, we rely on two proof ingredients.
The first one is that we use known algorithms for {\sc Vertex Cover} and {\sc Connected Vertex Cover} restricted to $H$-free graphs as subroutines in our new algorithms.
The second is that we consider the connected variant of the transversal problems in a more general form.
For {\sc Connected Vertex Cover} this variant is defined as follows:

\problemdef{{\sc Connected Vertex Cover Extension}}{a graph $G=(V,E)$, a subset $W\subseteq V$ and a positive integer $k$.}{does~$G$ have a connected vertex cover~$S_W$ with $W\subseteq S_W$ and $|S_W|\leq k$?}

\noindent
Note that {\sc Connected Vertex Cover Extension} becomes the original problem if $W=\emptyset$.
We define the problems {\sc Connected Feedback Vertex Set Extension} and {\sc Connected Odd Cycle Transversal Extension} analogously.
We will prove all our results for connected feedback vertex sets and connected odd cycle transversals for the extension versions.
These extension versions will serve as auxiliary problems for some of our inductive arguments, but this approach also leads to slightly stronger results.

\begin{remark2}\label{rem:1}
For any connected extension variant of these problems on ${\cal H}$-transversals, we may assume that the input graph~$G$ is connected.
If it is not, then either all but at most one connected component of~$G$ is ${\cal H}$-free and does not intersect~$W$, in which case it need not be considered, or the answer is immediately no.
It is easy to check ${\cal H}$-freeness for the three problems we consider.
\end{remark2}

\begin{sloppypar}
\begin{remark2}\label{rem:2}
Note that one could also define extension versions for any original transversal problem
(that is, where there is no requirement for the transversal to be connected).
However, such extension versions will be polynomially equivalent.
Indeed, we can solve the extension version on the input $(G,W,k)$ by considering the original problem on the input $(G-W,\max\{0, k-|W|\})$ and adding~$W$ to the solution.
However, due to the connectivity condition, we cannot use this approach for the connected variants.
\end{remark2}
\end{sloppypar}

\begin{remark2}\label{rem:3}
It is known that {\sc Vertex Cover} is polynomial-time solvable on $(P_1+\nobreak H)$-free graphs whenever this is the case on $H$-free graphs.
This follows from a well-known observation, see, for example,~\cite{Mo12}: one can solve the complementary problem of finding a maximum independent set in a $(P_1+\nobreak H)$-free graph by solving this problem on each $H$-free graph obtained by removing a vertex and all of its neighbours.
However, this trick does not work for {\sc Connected Vertex Cover}.
Moreover, it does not work for {\sc Feedback Vertex Set} and {\sc Odd Cycle Transversal} and their connected variants either.
\end{remark2}

\section{Preliminaries}\label{s-pre}

Let $G=(V,E)$ be a graph.
For a set $S\subseteq V$, we write~$G[S]$ to denote the subgraph of~$G$ induced by~$S$.
We say that~$S$ is \emph{connected} if~$G[S]$ is connected.
We write $G-S$ to denote the graph $G[V\setminus S]$.
A subset~$D\subseteq V$ is a \emph{dominating} set of~$G$ if every vertex of $V\setminus D$ is adjacent to at least one vertex of~$D$.
An edge~$uv$ of a graph $G=(V,E)$ is \emph{dominating} if $\{u,v\}$ is a dominating set.
The \emph{complement} of~$G$ is the graph $\overline{G}=(V,\{uv\; |\; uv\not \in E\; \mbox{and}\; u\neq v\})$.
The \emph{neighbourhood} of a vertex $u\in V$ is the set $N_G(u)=\{v\; |\; uv\in E\}$ and for $U\subseteq V$, we let $N_G(U)=\bigcup_{u\in U}N(u)\setminus U$.
We omit the subscript when there is no ambiguity.
We denote the \emph{degree} of a vertex $u\in V$ by $\deg(u)=|N_G(u)|$.

Let $G=(V,E)$ be a graph and let $S\subseteq V$.
Then~$S$ is a \emph{clique} if the vertices of~$S$ are pairwise adjacent and an \emph{independent set} if the vertices of~$S$ are pairwise non-adjacent.
A graph is \emph{complete} if its vertex set is a clique.
We let~$K_r$ denote the complete graph on~$r$ vertices.
Let $T\subseteq V$ with $S\cap T=\emptyset$.
Then~$S$ is \emph{complete} to~$T$ if every vertex of~$S$ is adjacent to every vertex of~$T$, and~$S$ is \emph{anti-complete} to~$T$ if there are no edges between~$S$ and~$T$.
In the first case, we also say that~$S$ is \emph{complete} to~$G[T]$ and in the second case \emph{anti-complete} to~$G[T]$.

A graph is \emph{bipartite} if its vertex set can be partitioned into at most two independent sets.
A bipartite graph is \emph{complete} bipartite if its vertex set can be partitioned into two independent sets~$X$ and~$Y$ such that
$X$ is complete to~$Y$. If $X$ or $Y$ has size~$1$, the complete bipartite graph is said to be a {\it star}.
Note that every edge of a complete bipartite graph is dominating.

Let~$G_1$ and~$G_2$ be two vertex-disjoint graphs.
The \emph{union} operation creates the \emph{disjoint union} $G_1+\nobreak G_2$ of~$G_1$ and~$G_2$, that is, the graph with vertex set $V(G_1)\cup V(G_2)$ and edge set $E(G_1)\cup E(G_2)$.
We denote the disjoint union of~$r$ copies of~$G_1$ by~$rG_1$.
The \emph{join} operation adds an edge between every vertex of~$G_1$ and every vertex of~$G_2$.
A graph~$G$ is a \emph{cograph} if~$G$ can be generated from~$K_1$ by a sequence of join and union operations.
A graph is a cograph if and only if it is $P_4$-free (see, for example,~\cite{BLS99}).

The following lemma is well known, but we include a short proof for completeness.

\begin{lemma}\label{l-p4}
Every connected $P_4$-free graph on at least two vertices has a spanning complete bipartite subgraph which can be found in polynomial time.
\end{lemma}

\begin{proof}
Let~$G$ be a connected $P_4$-free graph on at least two vertices.
Then $G$ is the join of two graphs~$G[X]$ and~$G[Y]$. Hence, $G$ has a spanning complete bipartite subgraph with partition classes~$X$ and~$Y$. Note that this implies that~$\overline{G}$ is disconnected. In order to find a (not necessarily unique) spanning complete bipartite subgraph of~$G$ with partition classes~$X$ and~$Y$ in polynomial time, we put the vertices of one connected component of~$\overline{G}$ in~$X$ and all the other vertices of~$\overline{G}$ in~$Y$.\qed
\end{proof}

\noindent
Grzesik et al.~\cite{GKPP19} gave a polynomial-time algorithm for finding a maximum independent set of a $P_6$-free graph in polynomial time. As the complement $V(G)\setminus I$ of every independent set~$I$ of a graph~$G$ is a vertex cover, their result implies that {\sc Vertex Cover} is polynomial-time solvable on $P_6$-free graphs.
Using the folklore trick mentioned in Remark~\ref{rem:3} (see also, for example,~\cite{JPP18,Mo12}) their result can also be formulated as follows.

\begin{theorem}[\cite{GKPP19}]\label{sp1p6-vc}
For every $s\geq 0$, {\sc Vertex Cover} can be solved in polynomial time on $(sP_1+\nobreak P_6)$-free graphs.
\end{theorem}

\noindent 
We recall also that {\sc Connected Vertex Cover} is polynomial-time solvable on $(sP_1+\nobreak P_5)$-free graphs~\cite{JPP18}.
We will need the extension version of this result.
Its proof is based on a straightforward adaption of the proof for {\sc Connected Vertex Cover} on $(sP_1+\nobreak P_5)$-free graphs~\cite{JPP18}.\footnote{See Appendix~\ref{a-cvc}, where we include a proof for reviewing purposes.}

\begin{theorem}[\cite{JPP18}]\label{t-cvc}
For every $s\geq 0$, {\sc Connected Vertex Cover Extension} can be solved in polynomial time on $(sP_1+\nobreak P_5$)-free graphs.
\end{theorem}

\section{The Case $\mathbf{H=P_4}$}\label{s-p4}

Recall that Brandst{\"a}dt and Kratsch~\cite{BK85} proved that {\sc Feedback Vertex Set} and {\sc Odd Cycle Transversal} can be solved in polynomial time on permutation graphs, which form a superclass of the class of $P_4$-free graphs.
Hence, we obtain the following proposition.

\begin{proposition}[\cite{BK85}]\label{p-oct}
{\sc Feedback Vertex Set} and {\sc Odd Cycle Transversal} can be solved in polynomial time on $P_4$-free graphs.
\end{proposition}

In this section, we prove that the (extension versions of the) connected variants of {\sc Feedback Vertex Set} and {\sc Odd Cycle Transversal} are also polynomial-time solvable on $P_4$-free graphs.
We make use of Proposition~\ref{p-oct} in the proofs.

\begin{theorem}\label{t-cfve}
{\sc Connected Feedback Vertex Set Extension} can be solved in polynomial time on $P_4$-free graphs.
\end{theorem}

\begin{proof}
Let $G=(V,E)$ be a $P_4$-free graph on~$n$ vertices and let~$W$ be a subset of~$V$.
By Remark~\ref{rem:1}, we may assume that~$G$ is connected.
By Lemma~\ref{l-p4}, in polynomial time we can find a spanning complete bipartite subgraph $G'=(X,Y,E')$, and we note that, by definition, every edge in~$G'$ is dominating.
Below, in Step~\ref{step1:1}, in polynomial time we compute a smallest connected feedback vertex set of~$G$ that contains~$W$ and intersects both~$X$ and~$Y$.
In Step~\ref{step1:2}, in polynomial time we compute a smallest connected feedback vertex set of~$G$ that contains $W$ and that is a subset of either~$X$ or~$Y$ (if such a set exists).
Then the smallest set found is a smallest connected feedback vertex set of~$G$ that contains~$W$.

\thmstep{\label{step1:1}Compute a smallest connected feedback vertex set~$S$ of $G$ such that $W \subseteq S$, $S\cap X \neq \emptyset$ and $S\cap Y\neq \emptyset$.}
We perform Step~\ref{step1:1} as follows. Consider two vertices $u \in X$ and $v \in Y$.
We shall describe how to find a smallest connected feedback vertex set of $G$ that contains $W \cup \{u,v\}$.
We find a smallest feedback vertex set~$S'$ in $G-(W\cup \{u,v\})$.
As $G-(W\cup \{u,v\})$ is $P_4$-free, this takes polynomial time by Proposition~\ref{p-oct}.
Then $S' \cup W\cup \{u,v\}$ is a smallest feedback vertex set of $G$ that contains $W \cup \{u,v\}$ and is connected, since~$uv$ is a dominating edge. By repeating this polynomial-time procedure for all~$O(n^2)$ possible choices of~$u$ and~$v$, we will find~$S$ in polynomial time.

\thmstep{\label{step1:2}Compute a smallest connected feedback vertex set~$S$ of $G$ such that $S \subseteq X$ or $S \subseteq Y$.}
For Step~\ref{step1:2} we describe only the $S \subseteq X$ case, as the $S\subseteq Y$ case is symmetric. Thus we may assume that $W \subseteq X$, otherwise no such set exists.
Clearly, we may also assume that~$G[Y]$ contains no cycles.
If~$G[Y]$ contains an edge it follows that $S=X$, otherwise $G-S$ would contain a triangle.
Suppose instead that~$Y$ is an independent set.
If $|Y|=1$, then $X \setminus S$ must be an independent set, otherwise $G-S$ contains a triangle.
So~$S$ is a smallest connected vertex cover of~$G[X]$ that contains~$W$.
As $G[X]$ is $P_4$-free, we can find such an~$S$ in polynomial time by Theorem~\ref{t-cvc}.
If $|Y|\geq 2$, then $|X\setminus S|\leq 1$, as otherwise $G-S$ contains a $4$-cycle.
Thus, we check, in polynomial time, if there exists a vertex $x\in X\setminus W$, such that $X\setminus \{x\}$ is connected.
If so, $S= X \setminus \{x\}$.\qed
\end{proof}

\begin{theorem}\label{t-cocte}
{\sc Connected Odd Cycle Transversal Extension} can be solved in polynomial time on $P_4$-free graphs.
\end{theorem}

\begin{proof}
We only provide an outline, as the proof follows that of Theorem~\ref{t-cfve}.
We perform the same two steps.
In Step~\ref{step1:1}, we need to find a smallest odd cycle transversal~$S'$ in $G-(W\cup \{u,v\})$ and can again apply Proposition~\ref{p-oct}.
In Step~\ref{step1:2}, we again note that if~$G[Y]$ contains an edge, then $S=X$.
Suppose that~$Y$ is an independent set. Then $G-S$ contains no odd cycles if and only if $X \setminus S$ is independent, so~$S$ is a smallest connected vertex cover of~$G[X]$ that contains~$W$.
(That is, the $|Y|=1$ case from the proof of Theorem~\ref{t-cfve} can be used for all values of~$|Y|$, as we are no longer concerned with whether $G-S$ might contain cycles of even length.)\qed
\end{proof}

\section{The Case $\mathbf{H=sP_1+\nobreak P_3}$}\label{s-p1p3}

\begin{sloppypar}
In this section, we will prove that {\sc Feedback Vertex Set} and {\sc Odd Cycle Transversal} and their connected variants can be solved in polynomial time on $(sP_1+\nobreak P_3)$-free graphs.
We need three structural results.
First, let us define a function~$c$ on the non-negative integers by $c(s):=\max\{3,2s-1\}$.
We will use this function~$c$ throughout the remainder of this section, starting with the following lemma.
\end{sloppypar}

\begin{lemma}\label{sp1p3-s}
Let $s \geq 0$ be an integer.
Let~$G$ be a bipartite $(sP_1+\nobreak P_3)$-free graph.
If~$G$ has a connected component on at least~$c(s)$ vertices, then there are at most~$s-1$ other connected components of~$G$ and each of them is on at most two vertices.
\end{lemma}

\begin{proof}
First note that the $s=0$ case of the lemma is trivially true, as every connected component of a bipartite $P_3$-free graph has at most two vertices.

Suppose, for contradiction, that~$G$ has a connected component~$C_1$ on at least $c(s)$ vertices and a connected component~$C_2$ on at least three vertices.
As~$C_1$ is bipartite and contains at least $2s-1$ vertices, $C_1$ contains a independent set of~$s$ vertices that induce~$sP_1$.
As~$C_2$ is bipartite and contains at least three vertices, $C_2$ has a vertex~$v$ of degree at least~$2$, and so~$v$ and two of its neighbours induce a~$P_3$.
Thus~$G$ is not $(sP_1+\nobreak P_3)$-free, a contradiction.

Similarly, if~$G$ contains a connected component~$C_1$ on at least $c(s)\geq 3$ vertices, then this component contains an induced~$P_3$.
Since~$G$ is $(sP_1+\nobreak P_3)$-free, $G$ can contain at most~$s-1$ connected components other than~$C_1$. \qed
\end{proof}

The \emph{internal vertices} and \emph{leaves} of a tree are the vertices of degree at least~$2$ and degree~$1$, respectively.

\begin{lemma}\label{l-tree}
Let $s \geq 0$ be an integer.
Let~$T$ be an $(sP_1+\nobreak P_3)$-free tree.
Then~$T$ has at most~$4s$ internal vertices.
\end{lemma}

\begin{proof}
Let~$U$ be the set of internal vertices of~$T$.
Suppose that $|U|\geq 4s+1\geq 1$.
We will show that this leads to a contradiction.
As a path with at least $4s+\nobreak 1$ internal vertices contains an induced $sP_1+\nobreak P_3$, we may assume that~$T$ is not a path and so has at least three leaves.
Hence $|V(T)|\geq 4s+4$.

Let~$X$ and~$Y$ be the two bipartition sets of~$T$, and assume without loss of generality that $|X|\geq 2s+2$.
For $Z \in \{X,Y\}$, let~$L_Z$ and~$U_Z$ be the leaves and internal vertices of~$T$ that belong to~$Z$.
If there is a vertex in~$Y$ of degree at least~$2$ that is anti-complete to a set of~$s$ vertices of~$X$, then~$T$ contains an induced $sP_1+\nobreak P_3$, a contradiction.
Therefore we may assume that every vertex of~$Y$ either has degree at least $|X|-s+1$ or is in~$L_Y$.
Then
\begin{eqnarray*}
|X|+|U_Y|+|L_Y|-1&=&|X|+|Y|-1 \\
&=&|V(T)|-1\\
&=& |E(T)| \\
&=&\sum_{v\in Y} \deg(v)\\
&\geq& \sum_{v\in U_Y} (|X|-s+1)+|L_Y|\\
&=& (|X|-s+1)|U_Y|+|L_Y| \\
&=&|X||U_Y|-s|U_Y|+|U_Y|+|L_Y|.
\end{eqnarray*}
Thus we have $|X|-1 \geq |X||U_Y|-s|U_Y|$ and we rearrange to see that $$|U_Y|\leq \frac{|X|-1}{|X|-s}=1+\frac{s-1}{|X|-s}.$$
Since $|X|\geq 2s+2$, we have that $|U_Y|<2$.
First suppose  $|U_Y|=0$. Then $|U_X|\leq 1$ 
and $|L_X|=0$, 
or $|U_X|=0$ and $|L_X|\leq 1$. 
Both cases contradict the assumption that~$X$ has at least $2s+2$ vertices.
Now suppose $|U_Y|=1$. Then, by our assumption that $|U|\geq 4s+1$, we have that $|U_X|\geq 4s$ and so $|L_Y|\geq |U_X|\geq 4s$.
Now it is easy to find an induced $sP_1+\nobreak P_3$ (see \figurename~\ref{intree}), and this contradiction completes the proof.\qed
\end{proof}

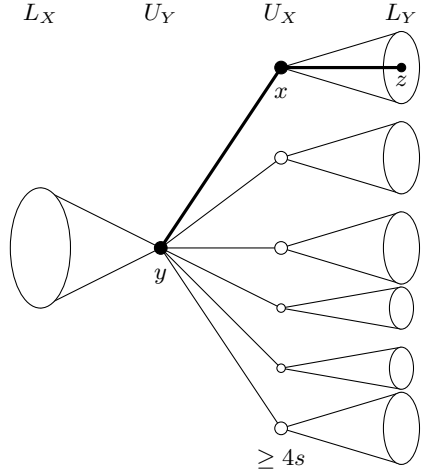
\begin{figure}
\begin{center}
\begin{tikzpicture}[scale=0.8]
\draw (-1,0)(1,3) -- (3,3.6) (1,3) -- (3,2.4) (-1,0) -- (1,1.5) -- (3,2.1) (1,1.5) -- (3,0.9) (-1,0) -- (1,0) -- (3,0.6) (1,0) -- (3,-0.6) (-1,0) -- (1,-3) -- (3,-3.6) (1,-3) -- (3,-2.4)
(-1,0) -- (1,-1) -- (3,-0.65) (1,-1) -- (3,-1.35) (-1,0) -- (1,-2) -- (3,-1.65) (1,-2) -- (3,-2.35) (-3,1) -- (-1,0) -- (-3,-1);
\draw [fill=white] (1,1.5) circle [radius=3pt] (1,0) circle [radius=3pt] (1,-1) circle [radius=2pt] (1,-2) circle [radius=2pt] (1,-3) circle [radius=3pt]
(-3,0) ellipse (0.5cm and 1cm) (3,3) ellipse (0.3cm and 0.6cm) (3,1.5) ellipse (0.3cm and 0.6cm) (3,0) ellipse (0.3cm and 0.6cm)
(3,-1) ellipse (0.2cm and 0.35cm) (3,-2) ellipse (0.2cm and 0.35cm) (3,-3) ellipse (0.3cm and 0.6cm);
\draw [fill=black] (-1,0) circle [radius=3pt] (1,3) circle [radius=3pt] (3,3) circle [radius=2pt];
\draw[very thick] (-1,0) -- (1,3) -- (3,3);
\node[below] at (-1,-0.2) {$y$};
\node[below] at (1,2.8) {$x$};
\node[below] at (3,3) {$z$};
\node[below] at (1,-3.2) {$\geq 4s$};
\node[above] at (-3,3.6) {$L_X$};
\node[above] at (-1,3.6) {$U_Y$};
\node[above] at (1,3.6) {$U_X$};
\node[above] at (3,3.6) {$L_Y$};
\end{tikzpicture}
\end{center}
\caption{The structure of the tree~$T$ in the proof of Lemma~\ref{l-tree} in the case when $|U_Y|=1$.
The set~$L_X$ is an independent set of vertices that each are adjacent to the unique vertex $y\in U_Y$.
The set $L_Y$ is partitioned into independent sets of vertices that have the same neighbour in~$U_X$.
The vertices $y,x,z$, together with~$s$ vertices of~$L_y$ not adjacent to~$x$, induced an $sP_1+\nobreak P_3$ in~$T$ 
(which leads to the desired contradiction in the proof).}
\label{intree}
\end{figure}

The bound of~$4s$ in Lemma~\ref{l-tree} is not tight but, as we shall see later, it suffices for our purposes.

\begin{lemma}  \label{l-rconnectsu}
Let $s \geq 0$ be an integer.  Let $G$ be a connected $(sP_1+P_3)$-free graph, and let $U$ be a set of vertices in $G$.  Then there is a set of vertices $R$ in $G$ such that $G[R \cup U]$ is connected and $|R| \leq 2s^2-2s+3$.
\end{lemma}

\begin{proof}
If~$G[U]$ is connected, then let $R = \emptyset$.  
Otherwise, since~$G$ cannot now be a complete graph, it contains an induced path~$P$ on three vertices in~$G$.
The number of connected  components of~$G[U]$ that do not contain a vertex that is either in~$P$ or adjacent to a vertex of~$P$ in~$G$ is at most $s-1$, otherwise~$G$ contains an induced~$sP_1+\nobreak P_3$.
Let~$R$ contain the vertices of~$P$ and the internal vertices of shortest paths in~$G$ from~$P$ to each set of vertices that induces a connected component of~$G[U]$.
As at most $s-1$ of these shortest paths have more than zero internal vertices, and as each contains at most~$2s$ internal vertices (any longer path contains an induced $sP_1+\nobreak P_3$), it follows that $|R| \leq 3+2s(s-1)=2s^2-2s+3$.
As~$G[R\cup U]$ is connected, the lemma is proved.
\qed
\end{proof}

We now prove our four results. For the connected variants, we consider the more general extension versions.

\begin{theorem}\label{sp1p3w-fvsi}
For every $s\geq 0$, {\sc Feedback Vertex Set} can be solved in polynomial time on $(sP_1+\nobreak P_3)$-free graphs.
\end{theorem}

\begin{proof}
Let $s \geq 0$ be an integer, and let $G=(V,E)$ be an $(sP_1+\nobreak P_3)$-free graph.
We must show how to find a smallest feedback vertex set of~$G$.
We will in fact show how to find a largest induced forest of~$G$, the complement of a smallest feedback vertex set.
The proof is by induction on~$s$.
If $s=0$, then we can use Proposition~\ref{p-oct}.
We now assume that $s\geq 1$ and that we have a polynomial-time algorithm for finding a largest induced forest in $((s-\nobreak 1)P_1+\nobreak P_3)$-free graphs.
Our algorithm performs the following two steps in polynomial time. Together, these two steps cover all possibilities.

\thmstep{\label{step2:1}Compute a largest induced forest~$F$ such that every connected component of~$F$ has at least~$c(s)$ vertices.}
By Lemma~\ref{sp1p3-s} we know that~$F$ will be connected, and so by Lemma~\ref{l-tree} $F$ will be a tree with at most~$4s$ internal vertices.
We consider every possible choice~$U$ of a non-empty set of at most~$4s$ vertices.
There are~$O(n^{4s})$ choices.
If~$U$ induces a tree, we will find a largest induced tree whose internal vertices all belong to~$U$.
This can be found by adding to~$U$ the largest possible set of vertices that are independent and belong to the set~$R$ of vertices in $G-U$ that each have exactly one neighbour in~$U$.
That is, we need a largest independent set in~$G[R]$ and, by Theorem~\ref{sp1p6-vc}, such a set can be found in polynomial time.

\thmstep{\label{step2:2}Compute a largest induced forest~$F$ such that~$F$ has a connected component with at most $c(s)-1$ vertices.}
We consider every possible choice of a non-empty set~$T$ of at most $c(s)-1$ vertices and discard those that do not induce a tree.
There are $O(n^{c(s)-1})$ choices for~$T$.
Let $U=N(T)$, and let $G'=G-(T\cup U)$.
Then~$G'$ is $((s-\nobreak 1)P_1+\nobreak P_3)$-free.
Thus we can find a largest induced forest~$F'$ of~$G'$ in polynomial time and $F' +\nobreak G[T]$ is a largest induced forest of~$G$ among those that have~$G[T]$ as a connected component.
\qed
\end{proof}

\begin{theorem}\label{sp1p3-cfvsi}
For every $s\geq 0$, {\sc Connected Feedback Vertex Set Extension} can be solved in polynomial time on $(sP_1+\nobreak P_3)$-free graphs.
\end{theorem}

\begin{proof}
There are similarities to the proof of Theorem~\ref{sp1p3w-fvsi}, but more arguments are needed.
Let $s \geq 0$ be an integer, let $G=(V,E)$ be a connected
 $(sP_1+\nobreak P_3)$-free graph and let~$W$ be a subset of~$V$.
We must show how to find a smallest connected feedback vertex set of~$G$ that contains~$W$ in polynomial time.
We show how to solve the complementary problem in polynomial time: how to find a largest induced forest~$F$ of~$G$ that does not include any vertex of~$W$ and $V \setminus F$ is connected.
We will say that an induced forest~$F$ is \emph{good} if it has these two properties.

Our algorithm performs the following three steps in polynomial time. Together, these three steps cover all possibilities.

\thmstep{\label{step3:1}Compute a largest good induced forest~$F$ such that there is a connected component of~$F$ that has at least~$c(s)$ vertices.}
By Lemma~\ref{sp1p3-s} we know that~$F$ has exactly one connected component on at least~$c(s)$ and there are at most $s-1$ other connected components of~$F$, each on at most two vertices.
By Lemma~\ref{l-tree}, the connected component on at least~$c(s)$ vertices has at most~$4s$ internal vertices.
We consider~$O(n^{4s+2(s-1)})$ choices of a non-empty set~$U$ of at most~$4s$ vertices that induces a tree and a set~$U'$ of at most $2(s-1)$ vertices that induces a disjoint union of vertices and edges such that $U \cup U'$ does not intersect~$W$,
 $U$ is disjoint from~$U'$ and no vertex of~$U$ has a neighbour in~$U'$.
Let~$R$ be the set of vertices that each have exactly one neighbour in~$U$ and no neighbour in~$U'$, but do not belong to~$W$.
We then add to~$U\cup U'$ the largest possible set~$L$ of vertices that are independent and belong to the set~$R$ such that $G-(L\cup U \cup U')$ is connected.
This is achieved by taking the complement of the smallest connected vertex cover of~$G-(U \cup U')$ that contains $V \setminus (R \cup U \cup U')$.
By Theorem~\ref{t-cvc}, this can be done in polynomial time.

\thmstep{\label{step3:2}Compute a largest good induced forest~$F$ such that~$F$ has at most $s-1$ connected components and each connected component has at most $c(s)-1$ vertices.}
Since the number of vertices in~$F$ is bounded by the constant $(s-1)(c(s)-1)$, we can simply check all sets containing at most that many vertices to see if they induce such a good forest.

\thmstep{\label{step3:3}Compute a largest good induced forest~$F$ such that~$F$ has at least~$s$ connected components and each connected component has at most $c(s)-1$ vertices.}
We consider $O(n^{s(c(s)-1)})$ choices of a non-empty set~$L$ of at most $s(c(s)-1)$ vertices.
We reject~$L$ unless~$G[L]$ is a good induced forest on~$s$ connected components with no connected component of more than $c(s)-1$ vertices.
Assuming our choice of~$L$ is correct, the connected components of~$G[L]$ will become connected components of~$G[F]$.

Let $U=N(L)$ and note that no vertex of~$U$ is in~$F$.
If $G-U$ is a good forest, then we are done.
Otherwise we consider every set~$R$ of at most $2s^2-2s+3$ vertices of $G - (L \cup U \cup W)$ such that $G[R \cup U \cup W]$ is connected; see also \figurename~\ref{f-step3}.
We note that if there is a largest induced forest~$F$ such that the connected components of~$G[L]$ are also connected components of~$G[F]$, then Lemma~\ref{l-rconnectsu} applied to $G-F$ implies that such a set~$R$ exists.

Let $S=R \cup U \cup W$.
If $G-S$ is a forest, then we are done.
Otherwise note that $G-(L \cup S)$ is the disjoint union of one or more complete graphs: $G-(L \cup S)$ cannot contain an induced~$P_3$, as it is anti-complete to~$L$ which contains an induced~$sP_1$.

As~$G$ is connected, each of the complete graphs in $G-(L\cup S)$ contains at least one vertex that is adjacent to some vertex of~$S$.
Hence in polynomial time we can find a set~$S'$ of vertices containing all but $\min\{2,|X|\}$ vertices from each of the complete graphs~$X$ in such a way that $G[S \cup S']$ is connected.
Then $G-(S \cup S')$ is a largest good induced forest that contains~$L$ 
and no vertex of $R\cup U$. 

\begin{figure}
\begin{center}
\begin{tikzpicture}[xscale=0.5, yscale=0.5]
\draw 
(-3.5,-8) rectangle (-2.5,-5)
(-3.4,-4.5) rectangle (-2.6,-2.5)
(-3.4,-2) rectangle (-2.6,-0)
(-3.5,0.5) rectangle (-2.5,3.5)
(-3.5,4.5) rectangle (-2.5,7.5)
(5,-7) rectangle (9,6.5)
(-2.5,5.7) -- (2,1)
(-2.5,2) -- (2,0)
(-2.6,-1) -- (2,-1)
(-2.6,-3.5) -- (2,-2)
(-2.5,-6) -- (2,-3)
(-4,8) -- (-4.5,8) -- (-4.5,-8.5) -- (-4,-8.5);
\draw[fill=gray!30!white] (2,-2) ellipse (2cm and 6cm);
\draw[fill=gray] (2.5,5) ellipse (1cm and 2cm);
\node[below] at (7,-7) {$G-(L\cup U\cup W)$};
\node[above] at (2.5,7) {$W$};
\node[below] at (-3,-8) {$L$};
\node[below] at (2,-8) {$U=N(L)$};
\node[left] at (-5,0) {$s$};
\end{tikzpicture}
\end{center}
\caption{The decomposition of the $(sP_1+\nobreak P_3)$-free graph~$G$, as given in Step~\ref{step3:3} of the algorithm from the proof of Theorem~\ref{sp1p3-cfvsi}.}
\label{f-step3}
\end{figure}
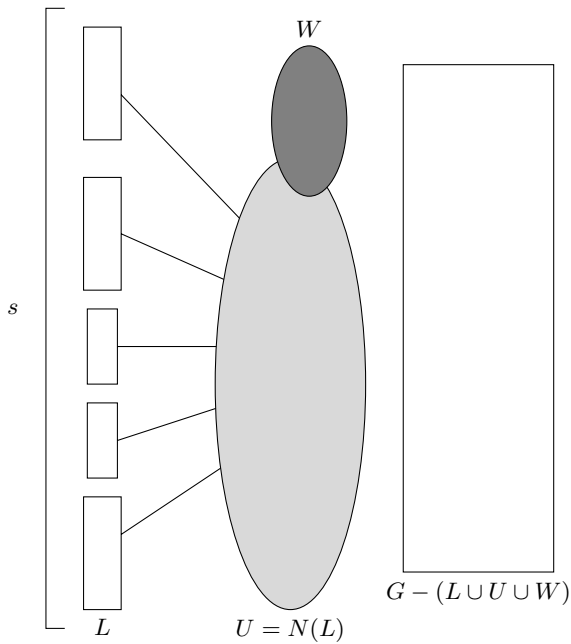

\noindent
After considering each of the $O(n^{2s^2-2s+3})$ choices for~$R$, in polynomial time we find a largest good induced forest that contains~$L$ and no vertex of $U$.
After considering each of the $O(n^{s(c(s)-1)})$ choices for~$L$, we find in polynomial time a largest good induced forest that has at least~$s$ connected components, each with at most $c(s)-1$ vertices.
\qed
\end{proof}

\newpage
\begin{theorem}\label{sp1p3w-octi}
For every $s\geq 0$, {\sc Odd Cycle Transversal} can be solved in polynomial time on $(sP_1+\nobreak P_3)$-free graphs.
\end{theorem}

\begin{proof}
Let $s \geq 0$ be an integer, and let $G=(V,E)$ be an $(sP_1+\nobreak P_3)$-free graph.
We must describe how to find a smallest odd cycle transversal of~$G$.
If $s=0$, then we can use Proposition~\ref{p-oct}.
We now assume that $s\geq 1$ and use induction.
We will in fact describe how to solve the complementary problem and find a largest induced bipartite subgraph of~$G$.
The proof is by induction on~$s$ and our algorithm performs two steps in polynomial time, which together cover all possibilities.

\thmstep{\label{step4:1}Compute a largest induced bipartite subgraph~$B$ such that every connected component of~$B$ has at least~$c(s)$ vertices.}
By Lemma~\ref{sp1p3-s}, we know that~$B$ will be connected. Hence, $B$ has a unique bipartition, which we denote~$\{X,Y\}$.
We first find a largest induced bipartite subgraph~$B$ that is a star: we consider each vertex~$x$ and find a largest induced star centred at~$x$ by finding a largest independent set in~$N(x)$.
This can be done in polynomial time by Theorem~\ref{sp1p6-vc}.

Next, we find a largest induced bipartite subgraph~$B$ that is not a star.
We consider each of the~$O(n^2)$ choices of edges~$xy$ of~$G$ and find a largest induced connected bipartite subgraph~$B$ such that $x \in X$ and $y \in Y$ and neither~$x$ nor~$y$ has degree~$1$ in~$B$ (since~$B$ is not a star, it must contain such a pair of vertices).
Note that the number of vertices in~$X$ non-adjacent to~$y$ is at most $s-1$, otherwise~$B$ induces an $sP_1+\nobreak P_3$.
Similarly there are at most $s-1$ vertices in~$Y$ non-adjacent to~$x$.
We consider each of the~$O(n^{2s-2})$ possible pairs of disjoint sets~$X'$ and~$Y'$, which are each independent sets of size at most $s-1$ such that $X'\cup Y'$ is anti-complete to~$\{x,y\}$.
We will find a largest induced bipartite subgraph with partition classes~$X$ and~$Y$ such that $\{x\} \cup X' \subseteq X$ and $\{y\} \cup Y' \subseteq Y$ and every vertex in $X \setminus X'$ is adjacent to~$y$ and every vertex in $Y \setminus Y'$ is adjacent to~$x$.
That is, we must find a largest independent set in both $N(x) \setminus N(\{y\} \cup Y')$ and $N(y) \setminus N(\{x\} \cup X')$; see \figurename~\ref{f-il} for an illustration.
This can be done in polynomial time, again by applying Theorem~\ref{sp1p6-vc}.

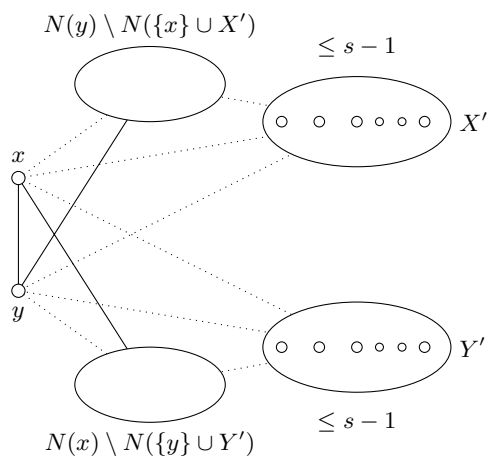
\begin{figure}
\begin{center}
\begin{tikzpicture}[xscale=0.5, yscale=0.5]
\draw (-1,-4) -- (-4.5,1.5) -- (-4.5,-1.5) -- (-1,4);
\draw[dotted] (4.5,3) -- (-4.5,1.5) -- (4.5,-3) (-4.5,1.5) -- (-1,4) -- (4.5,3)
(4.5,3) -- (-4.5,-1.5) -- (4.5,-3) (-4.5,-1.5) -- (-1,-4) -- (4.5,-3);
\draw[fill=white] 
(-4.5,1.5) circle [radius=5pt]
(-4.5,-1.5) circle [radius=5pt]
(2.5,3) circle [radius=4pt]
(3.5,3) circle [radius=4pt]
(4.5,3) circle [radius=4pt]
(5.1,3) circle [radius=3pt]
(5.7,3) circle [radius=3pt]
(6.3,3) circle [radius=4pt]
(2.5,-3) circle [radius=4pt]
(3.5,-3) circle [radius=4pt]
(4.5,-3) circle [radius=4pt]
(5.1,-3) circle [radius=3pt]
(5.7,-3) circle [radius=3pt]
(6.3,-3) circle [radius=4pt]
(-1,4) ellipse (2cm and 1cm)
(-1,-4) ellipse (2cm and 1cm)
(4.5,3) ellipse (2.5cm and 1.2cm)
(4.5,-3) ellipse (2.5cm and 1.2cm);
\node[above] at (4.5,4.5) {$\leq s-1$};
\node[below] at (4.5,-4.5) {$\leq s-1$};
\node[above] at (-1,5) {$N(y)\setminus N(\{x\}\cup X')$};
\node[below] at (-1,-5) {$N(x)\setminus N(\{y\}\cup Y')$};
\node[right] at (7,3) {$X'$};
\node[right] at (7,-3) {$Y'$};
\node[above] at (-4.5,1.7) {$x$};
\node[below] at (-4.5,-1.7) {$y$};
\end{tikzpicture}
\end{center}
\caption{An illustration of Step~\ref{step4:1} of the algorithm in the proof of Theorem~\ref{sp1p3w-octi}. Full and dotted lines indicate when two sets are complete or anti-complete to each other, respectively. The absence of a full or dotted lines indicates that edges may or may not exist between two sets.}
\label{f-il}
\end{figure}

\thmstep{\label{step4:2}Compute a largest induced bipartite subgraph~$B$ such that~$B$ has a connected component with at most $c(s)-1$ vertices.}
We consider each of the~$O(n^{c(s)-1})$ possible choices of a non-empty set~$L$ of at most $c(s)-1$ vertices and discard those that do not induce a bipartite graph.
We will find the largest~$B$ that has~$G[L]$ as a connected component.
Let $U=N(L)$, and let $G'=G-(L\cup U)$.
As~$G'$ is $((s-\nobreak 1)P_1+\nobreak P_3)$-free, we can find a largest induced bipartite subgraph~$B'$ of~$G'$ in polynomial time and $B' + G[L]$ is a largest induced bipartite subgraph among those that have~$G[L]$ as a connected component.
\qed
\end{proof}

\begin{theorem}\label{sp1p3-cocti}
For every $s\geq 0$, {\sc Connected Odd Cycle Transversal Extension} can be solved in polynomial time on $(sP_1+\nobreak P_3)$-free graphs.
\end{theorem}

\begin{proof}
Let $s \geq 0$ be an integer, let $G=(V,E)$ be a connected $(sP_1+\nobreak P_3)$-free graph and let~$W$ be a subset of~$V$.
We must describe how to find a smallest connected odd cycle transversal of~$G$ that contains~$W$.
We will solve the complementary problem: how to find a largest induced bipartite graph of~$G$ that does not include any vertex of~$W$ and whose complement is connected.
We will say that an induced bipartite graph~$B$ is \emph{good} if it has these two properties.
Our algorithm consists of three steps, which can each be performed in polynomial time and which together cover all the possible cases.
  
\thmstep{\label{step5:1}Compute a largest good induced bipartite subgraph~$B$ such that~$B$ has a bipartition $\{X,Y\}$ in which one set, say~$X$, has size $|X| \leq s$.
(Note that this includes the case when every connected component of~$B$ has at most two vertices and~$B$ has at most~$s$ connected components.)}
We consider $O(n^s)$ choices of an independent set~$X$ of at most~$s$ vertices of~$G$ that does not intersect~$W$.
We wish to find $Y$, the largest possible independent set in $G - (W \cup X)$ such that $G-(X \cup Y)$ is connected.
By Theorem~\ref{t-cvc}, we can do this in polynomial time by computing a minimum connected vertex cover of $G-X$ that contains~$W$ and taking its complement (in $G-X$).

\thmstep{\label{step5:4}Compute a largest good induced bipartite subgraph~$B$ such that~$B$ has at least~$s$ connected components and each connected component has at most two vertices.}
Note that $2 \leq c(s)-1$.
The algorithm mimics Step~\ref{step3:3} of the algorithm in the proof of Theorem~\ref{sp1p3-cfvsi}, but checks for a good bipartite graph instead of a good forest.

\thmstep{\label{step5:2}Compute a largest good induced bipartite subgraph~$B$ such that there is a connected component of~$B$ that has at least three vertices and~$B$ has a bipartition $\{X,Y\}$ with $|X| \geq s+1$ and $|Y|\geq s+1$.}
It is in this case that we must do most of the work in proving the theorem, and here we will need ideas beyond those already met in this section.  

As~$B$ contains a connected component on at least three vertices, it will contain an induced~$P_3$ and so $|X| \geq 1$ and $|Y| \geq 1$.
We consider $O(n^{2s+2})$ choices of disjoint independent sets~$X'$ and~$Y'$ that each contain $s+1$ vertices of~$G$ and do not intersect~$W$.
If $G[X' \cup Y']$ contains an induced $P_3$, our aim is to compute a largest good induced bipartite graph~$B$ with bipartition~$\{X,Y\}$ such that $X'\subseteq X$ and $Y' \subseteq Y$; otherwise we discard the choice of $X',Y'$.
 
We define (see also \figurename~\ref{fig}) a partition of $V \setminus (X' \cup Y')$:

\begin{gather*}
U = (N(X')\cap N(Y')) \cup W \\
V_X = N(X') \setminus (Y' \cup N(Y')\cup W) \\
V_{Y} = N(Y') \setminus (X' \cup N(X')\cup W) \\
Z = V \setminus (X' \cup Y' \cup N(X') \cup N(Y')\cup W)
\end{gather*}

There are a number of steps where our procedure branches as we consider all possible ways of choosing whether or not to add certain vertices to~$B$.
Note that assuming our choice of~$X'$ and~$Y'$ is correct, no vertex of~$U$ can be in~$B$.
If we decide that a vertex will not be in~$B$, we will then add it to~$U$.

\begin{figure}
\begin{center}
\begin{tikzpicture}[xscale=0.5, yscale=0.5]
\draw (-8.5,2.5) -- (-9,2.5) -- (-9,-3.5) -- (-8.5,-3.5);
\draw (-5.67,-0) -- (-2.72,2) (-1.72,-1.61) -- (-1.72,1.61) (-5.55,-1.62) -- (-2.72,-2) (-0.36,-1.57) -- (3.1,2.2) (-0.36,1.57) -- (3.1,-2.2);
\draw[dotted] (1,2.5) -- (2.55,2.8) (1,-2.5) -- (2.55,-2.8) (0.8,2) -- (8.5,0.5) (0.8,-2) -- (8.5,-0.5);
\draw[fill=white] (4,2.75) circle [radius=15pt] (5,3.5) circle [radius=12pt] (5.3,2.5) circle [radius=8pt] (6.5,3) circle [radius=15pt] (3.2,-3) circle [radius=8pt]
(4.1,-3.5) circle [radius=15pt] (4.5,-2.5) circle [radius=12pt] (5.3,-3.7) circle [radius=7pt] (5.7,-2.7) circle [radius=17pt] (6.5,-3.5) circle [radius=10pt];
\draw (-1,2.5) ellipse (2cm and 1cm);
\draw (-1,-2.5) ellipse (2cm and 1cm);
\draw (5,3) ellipse (2.5cm and 1.2cm);
\draw (5,-3) ellipse (2.5cm and 1.2cm);
\draw (-7,-1) ellipse (1.5cm and 2.2cm);
\draw[fill=gray] (-7,1.5) ellipse (0.5cm and 1cm);
\draw (10,0) ellipse (1.5cm and 3cm);
\node[left] at (-9,0) {$U$};
\node[above] at (-1,3.5) {$X'$};
\node[below] at (-1,-3.5) {$Y'$};
\node[above] at (4,4) {$V_Y$};
\node[below] at (4,-4) {$V_X$};
\node[right] at (11.5,0) {$Z$};
\node[above] at (-7,2.5) {$W$};
\node[below] at (-7,-3.5) {$N(X')\cap N(Y')$};
\end{tikzpicture}
\end{center}
\caption{The decomposition of~$G$ in Step~\ref{step5:2}.
Full and dotted lines indicate when two sets are complete or anti-complete to each other, respectively.
The absence of a full or dotted lines indicates that edges may or may not exist between two sets.
The circles in~$V_X$ and~$V_Y$ represent disjoint unions of complete graphs.}
\label{fig}
\end{figure}
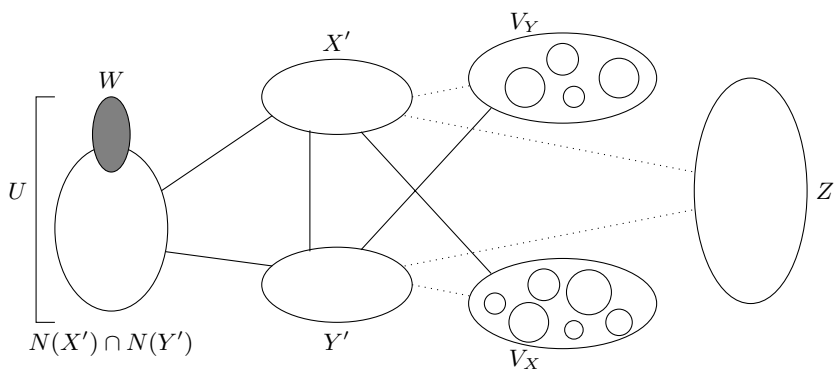

\thmsubstep{Reduce~$Z$ to the empty set.}
Notice that~$Z$ does not contain an independent set on more than $s-1$ vertices otherwise $G[X' \cup Y' \cup Z]$ would contain an induced $sP_1+\nobreak P_3$.
We consider $O(n^{2s-2})$ choices of disjoint independent sets~$Z_X$ and~$Z_Y$ that are each subsets of~$Z$ and each contain at most $s-1$ vertices.
We move the vertices of~$Z_X$ and~$Z_Y$ by adding them to~$X'$ and~$Y'$, respectively.
We move the vertices of $Z \setminus (Z_X \cup Z_Y)$ by adding them to~$U$.
If after this process is complete there are vertices in $V_X \cup V_Y$ with neighbours in both~$X'$ and~$Y'$, we move these vertices by adding them to~$U$.
We note that now:
\begin{itemize}
\item $Z$ is the empty set,
\item $V_X$ still contains vertices with neighbours in~$X'$ but not in~$Y'$,
\item $V_Y$ still contains vertices with neighbours in~$Y'$ but not in~$X'$, and
\item $U$ contains vertices that will not be in~$B$.
\end{itemize}
So our task is to decide how best to add vertices of~$V_X$ to~$Y'$ and vertices of~$V_Y$ to~$X'$, but first there is another step: as $G-B$ must be connected, and~$G[U]$ is a subgraph of~$G-B$, we choose some vertices that will not be in~$B$, but will connect together the connected components of~$G[U]$.
This will not be possible if the vertices of~$U$ belong to more than one connected component of $G-(X' \cup Y')$.
Hence, in that case we discard this choice of $Z_X,Z_Y$.

\thmsubstep{Make $G[U]$ connected.}
We consider $O(n^{2s^2-2s+3})$ choices of sets~$R$ of vertices of $G-(X' \cup Y')$ such that each contains at most $2s^2-2s+3$ vertices.
If $G[R \cup U]$ is connected, we move the vertices of~$R$ by adding them to~$U$, and so~$G[U]$ becomes connected.
Note that since all vertices of~$U$ are in the same connected component of $G-(X' \cup Y')$, Lemma~\ref{l-rconnectsu} implies that at least one such set~$R$ can be found.

\thmsubstep{Add vertices from $V_X$ to $Y'$ and from $V_Y$ to $X'$.}
We note that~$G[V_X]$ is $P_3$-free, as no vertex of~$V_X$ has a neighbour in~$Y'$, $|Y'|\geq s$, and~$G$ is $(sP_1+\nobreak P_3)$-free.
By symmetry, $G[V_Y]$ is $P_3$-free.
Thus both~$G[V_X]$ and~$G[V_Y]$ are disjoint unions of complete graphs.
Note that~$B$ can contain at most one vertex from each of these complete graphs.
We consider two subcases.

\thmsubsubstep{Compute a largest good induced bipartite subgraph~$B$ with bipartition $\{X,Y\}$ such that $X' \subseteq X$, $Y' \subseteq Y$ and $G-B$ contains no edges between~$V_X$ and~$V_Y$.}
As $G-B$ must be connected, each clique of~$V_X$ and~$V_Y$ that contains at least two vertices must contain a vertex adjacent to~$U$ (otherwise such a set~$B$ cannot exist).
Thus we can form~$X$ from~$X'$ by adding to~$X'$ one vertex from each clique of~$V_Y$ and form~$Y$ by adding to~$Y'$ one vertex from each clique of~$V_X$ in such a way that~$G-B$ is connected.
(If we do this, it is possible that $G-B$ will contain an edge from~$V_X$ to~$V_Y$, but then this solution is at least as large as one where such edges are avoided.)

\thmsubsubstep{Compute a largest good induced bipartite subgraph~$B$ with bipartition $\{X,Y\}$ such that $X' \subseteq X$, $Y' \subseteq Y$ and $G-B$ has an edge $xy$ where $x \in V_X$, $y \in V_Y$.}
We consider $O(n^2)$ choices of an edge $xy$, $x \in V_X$, $y \in V_Y$.
Let $v_X \in X'$ be a neighbour of~$x$ and note that $v_X$, $x$ and $y$ induce a $P_3$ in~$G$.
Therefore~$x$ must be complete to all but at most $s-1$ cliques of~$V_Y$.
By symmetry, $y$ must be complete to all but at most $s-1$ cliques of~$V_X$.
A clique in~$V_X$ or~$V_Y$ is \emph{bad} if it is not complete to~$y$ or~$x$, respectively.
Note that the cliques containing~$x$ and~$y$ may be bad.
We move~$x$ and~$y$ to~$U$.

We consider $O(n^{2s-2})$ choices of a set~$S$ of at most $2s-2$ vertices that each belong to a distinct bad clique and move each to~$X'$ or~$Y'$ if they are in~$V_Y$ or~$V_X$ respectively.
We move the other vertices of the bad cliques to~$U$.
If the vertices of~$U$ are not in the same connected component of $G-(X' \cup Y')$, we discard this choice of~$S$.
We consider $O(n^{2s^2-2s+3})$ choices of sets~$R'$ of vertices of $G-(X' \cup Y')$ such that each contains at most $2s^2-2s+3$ vertices.
If $G[R' \cup U]$ is connected we move the vertices of~$R'$ to~$U$, so~$G[U]$ becomes connected.
Since the vertices of~$U$ are in the same connected component of $G-(X' \cup Y')$, Lemma~\ref{l-rconnectsu} implies that at least one such set~$R'$ can be found.

Note that some cliques might have been completely removed from~$V_X$ and~$V_Y$ by the choice of~$R'$.
It only remains to pick one vertex from each remaining clique of~$V_X$ and~$V_Y$, and add these vertices to~$Y'$ or~$X'$, respectively to finally obtain~$B$.
As all vertices in these cliques are adjacent to~$x$ or~$y$ we know that $G-B$ will be connected.
\qed
\end{proof}

\section{The Case $\mathbf{H=P_6}$}\label{s-hard}

In this section we prove that {\sc Odd Cycle Transversal} and {\sc Connected Odd Cycle Transversal} are \NP-hard on 
 $(P_2+\nobreak P_5,P_6)$-free graphs. We do this by modifying the construction used in~\cite{OR19} for proving that these two problems are \NP-complete on  $P_{13}$-free segment graphs. 

\begin{theorem}\label{thm:oct-P6-P2+P5}
{\sc Odd Cycle Transversal} and {\sc Connected Odd Cycle Transversal} are \NP-complete on $(P_2+\nobreak P_5,P_6)$-free graphs.
\end{theorem}

\begin{proof}
Both problems are readily seen to belong to \NP.
To prove \NP-hardness we reduce from {\sc Vertex Cover}, which is known to be \NP-complete~\cite{GJS76}.
Let $(G,k)$ be an instance of {\sc Vertex Cover}.
Let~$n$ and~$m$ be the number of vertices and edges, respectively, in~$G$.
Let $v_1,\ldots,v_n$ be the vertices of~$G$.
We construct a graph~$G^*$ from~$G$ as follows.
\begin{enumerate}
\item For $i \in \{1,\ldots,n\}$ create vertices $a_i,b_i,c_i,x_i$ and~$y_i$.
Let $A, B, C, X$ and $Y$ be the sets of, respectively, $a_i$, $b_i$, $c_i$, $x_i$ and $y_i$ vertices. 
\item For $i,j \in \{1,\ldots,n\}$, add the edges~$x_iy_j$ and~$b_iy_j$ (so we make $Y$ complete to both $X$ and $B$).
\item For each $i \in \{1,\ldots,n\}$, add edges $x_ia_i,x_ib_i,a_ib_i,b_ic_i,c_iy_i$ (a \emph{vertex gadget}, see also \figurename~\ref{fig:oct-vertex-gadget} and note that~$b_i$ is adjacent to~$y_i$ by the previous step).
\item For each edge~$v_iv_j$ in~$G$ with $i<j$, add a vertex~$d_{i,j}$ adjacent to both~$x_i$ and~$y_j$ (an \emph{edge gadget}, see also \figurename~\ref{fig:oct-edge-gadget}).
Let~$D$ be the set of~$d_{i,j}$ vertices.
\end{enumerate}

\begin{figure}[h]
\begin{center}
\begin{subfigure}[t]{0.3\textwidth}
\begin{center}
\begin{tikzpicture}
\node[circle, draw,label=below:$x_i$] (x) at (0,0) {};
\node[circle,draw,label=below:$y_i$] (y) at (2,0) {};
\node[circle,draw,label=above:$a_i$] (a) at (0,1) {};
\node[circle,draw,label=above:$b_i$] (b) at (1,1) {};
\node[circle,draw,label=above:$c_i$] (c) at (2,1) {};
\draw (x) -- (y);
\draw (x) -- (a) -- (b) -- (c) -- (y);
\draw (x) -- (b) -- (y);
\end{tikzpicture}
\caption{Vertex gadget}
\label{fig:oct-vertex-gadget}
\end{center}
\end{subfigure}
\begin{subfigure}[t]{0.3\textwidth}
\begin{center}
\begin{tikzpicture}
\node[circle,draw,label=below:$x_i$] (x) at (0,0) {};
\node[circle,draw,label=below:$y_j$] (y) at (2,0) {};
\node[circle,draw,label=above:$d_{i,j}$] (e) at (1,1) {};
\draw (x) -- (y) -- (e) -- (x);
\end{tikzpicture}
\caption{Edge gadget}
\label{fig:oct-edge-gadget}
\end{center}
\end{subfigure}
\end{center}
\caption{The two gadgets used in the proof of Theorem~\ref{thm:oct-P6-P2+P5}.}
\end{figure}
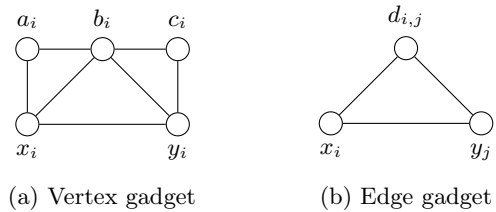

We first claim that the following statements are equivalent:
\begin{enumerate}[(i)]
\renewcommand{\theenumi}{(\roman{enumi})}
\renewcommand{\labelenumi}{(\roman{enumi})}
\item\label{it:i}$G$ has a vertex cover of size at most~$k$;
\item\label{it:ii}$G^*$ has an odd cycle transversal of size at most $n+\nobreak k$;
\item\label{it:iii}$G^*$ has a connected odd cycle transversal of size at most $n+\nobreak k$.
\end{enumerate}
The implication~\ref{it:iii} $\Rightarrow$~\ref{it:ii} is trivial. Below we prove~\ref{it:i} $\Rightarrow$~\ref{it:iii} and~\ref{it:ii} $\Rightarrow$~\ref{it:i}.

\medskip
\noindent
\ref{it:i}~$\Rightarrow$~\ref{it:iii}.
Suppose that $G$ has a vertex cover~$Q$ of size at most~$k$. We define the set
$$S = \bigcup_{v_i \in Q} \{x_i,y_i\} \cup \bigcup_{v_i \notin Q} \{b_i\}$$
and observe that $|S|=2|Q|+(n-|Q|)=n+|Q| \leq n+k$ and that $S$ is connected. We claim that $S$ is an odd cycle transversal of $G^*$. This can be seen as follows.
The only induced odd cycles in~$G^*$ are the three triangles in each vertex gadget and the triangle in each edge gadget.
By construction of~$S$, for every $i\in \{1,\ldots,n\}$, either~$S$ contains both~$x_i$ and~$y_i$ or~$S$ contains~$b_i$, thus every triangle in every vertex gadget intersects~$S$.
Furthermore, since~$Q$ is a vertex cover of~$G$, for every edge gadget $\{x_i,y_j,d_{i,j}\}$, either $x_i \in S$ or $y_j \in S$.
Therefore~$S$ intersects every odd cycle in~$G^*$.

\medskip
\noindent
\ref{it:ii}~$\Rightarrow$~\ref{it:i}.
Suppose that~$G^*$ has an odd cycle transversal~$S$ of size at most~$n+k$.
Consider an edge gadget on $\{x_i,y_j,d_{i,j}\}$.
If $d_{i,j} \in S$ then $S':=(S \setminus \{d_{i,j}\}) \cup \{x_i\}$ is an odd cycle transversal of~$G$ with $|S'| \leq |S|$.
We may therefore assume that~$S$ contains no vertices of~$D$.
For $i \in \{1,\ldots,n\}$, the vertex~$b_i$ intersects all odd cycles in the vertex gadget on $\{a_i,b_i,c_i,x_i,y_i\}$.
If $b_i \notin S$ then $|S \cap \{a_i,b_i,c_i,x_i,y_i\}| \geq 2$ since~$S$ intersects all induced odd cycles of the vertex gadget.
Note that~$\{x_i,y_i\}$ intersects all odd cycles of the vertex gadget.
Therefore, if $|S \cap \{a_i,b_i,c_i,x_i,y_i\}| \geq 2$, then $S':=(S \setminus \{a_i,b_i,c_i\}) \cup \{x_i,y_i\}$ is an odd cycle transversal of~$G^*$ with $|S'| \leq |S|$.
We may therefore assume that for every $i \in \{1,\ldots,n\}$, either $b_i \in S$ or $\{x_i,y_i\} \subseteq S$ and there are no other vertices in~$S$. Let $B_S=B\cap S$, $X_S=S\cap X$ and $Y_S= S\cap Y$. Then $|S|=|B_S|+|S_X|+|S_Y|=n+|S_X|$.
Let $Q=\bigcup_{x_i \in S}\{v_i\}$. Then $|Q|=|S_X|=|S|-n\leq n+k-n = k$.

We claim that $Q$ is a vertex cover of $G$. This can be seen as follows.
Consider an edge~$v_iv_j$ of~$G$ (without loss of generality assume $i<j$). 
Then $|\{x_i,y_j,d_{i,j}\} \cap S| \geq 1$, as~$S$ is an odd cycle transversal of~$G^*$.
By assumption on~$S$, $d_{i,j} \notin S$ and if $y_j \in S$ then $x_j \in S$.
It follows that $x_i \in S$ or $x_j \in S$ and so $v_i \in Q$ or $v_j \in Q$.
We conclude that~$Q$ is a vertex cover of~$G$ of size at most~$k$.

\medskip
\noindent
It only remains to show that~$G^*$ is $(P_2+\nobreak P_5,P_6)$-free.
Suppose, for contradiction, that $H \in \{P_2+\nobreak P_5,P_6\}$ is an induced subgraph of~$G^*$.
Every vertex in $A \cup C \cup D$ has degree~$2$ and its two neighbours are adjacent.
Therefore no vertex in $V(H) \cap (A \cup C \cup D)$ is an internal vertex of a path of~$H$. That is, if $x \in V(H) \cap (A \cup C \cup D)$ then~$x$ has degree~$1$ in~$H$.
Furthermore, $A \cup C \cup D$ is an independent set in~$G^*$. Hence, if $H=P_2+\nobreak P_5$, then at most one vertex of the~$P_2$ connected component of~$H$ can be in $A \cup C \cup D$.
We conclude that $G^*[V(H) \cap (B \cup X \cup Y)]$ contains an induced subgraph~$H'$ on four vertices that is isomorphic to~$P_1+\nobreak P_3$ if $H=P_2+P_5$ or~$P_4$ if $H=P_6$. Since~$Y$ is an independent set and~$B \cup X$ is a perfect matching, $H'$ must contain at least one vertex of~$B \cup X$ and at least one vertex of~$Y$.
As~$Y$ is complete to~$B \cup X$, we find that $H'$ contains either~$C_4$ or~$K_{1,3}$ as a (not necessarily induced) subgraph, a contradiction. This completes the proof.\qed
\end{proof}

The proof of  Theorem \ref{thm:oct-P6-P2+P5} gives a slightly stronger result if we assume the Exponential Time Hypothesis (ETH). The ETH is one of standard assumptions in complexity theory which, along with the \emph{sparsification lemma}, implies that {\sc $3$-Sat} with~$n$ variables and~$m$ clauses cannot be solved in~$2^{o(n+m)}$ time~\cite{IP01,IPZ01}.
The number of vertices in the graph~$G^*$ constructed in the proof of Theorem~\ref{thm:oct-P6-P2+P5} is $5n+\nobreak m$.
Thus an algorithm solving ({\sc Connected}) {\sc Odd Cycle Transversal} on $(P_2+\nobreak P_5,P_6)$-free graphs with~$n$ vertices in time~$2^{o(n)}$ could be used to solve {\sc Vertex Cover} on graphs with~$n$ vertices and~$m$ edges in~$2^{o(n+m)}$ time.
However, such a fast algorithm for {\sc Vertex Cover} does not exist unless the ETH fails~\cite{CFKLMM15}.
Thus we get the following statement.

\begin{corollary}
{\sc Odd Cycle Transversal} and {\sc Connected Odd Cycle Transversal} cannot be solved in~$2^{o(n)}$ time on $(P_2+\nobreak P_5,P_6)$-free graphs with~$n$ vertices, unless the ETH fails.
\end{corollary}

\section{Conclusions}\label{s-con}

We proved polynomial-time solvability of {\sc Feedback Vertex Set} and {\sc Odd Cycle Transversal} on $H$-free graphs when $H=sP_1+\nobreak P_3$ and polynomial-time solvability of their connected variants on $H$-free graphs, when $H=P_4$ or $H=sP_1+\nobreak P_3$; 
see also Table~\ref{t-thetable}, where we place these results in the context of known results for these problems on $H$-free graphs.
We also showed that {\sc Odd Cycle Transversal} and {\sc Connected Odd Cycle Transversal} are \NP-complete on $(P_2+\nobreak P_5,P_6)$-free graphs.

Natural cases for future work are the cases when $H=sP_1+\nobreak P_4$ for $s\geq 1$ and $H=P_5$ for all four problems (in particular the case when $H=P_5$ is the only open case for {\sc Odd Cycle Transversal} and {\sc Connected Odd Cycle Transversal} restricted to $P_r$-free graphs).
Note that Lemma~\ref{sp1p3-s} does not hold on $(sP_1+\nobreak P_4)$-free graphs: the disjoint union of any number of arbitrarily large stars is even $P_4$-free.

Recall that {\sc Vertex Cover} and {\sc Connected Vertex Cover} are polynomial-time solvable even on $(sP_1+\nobreak P_6)$-free graphs~\cite{GKPP19} and $(sP_1+\nobreak P_5)$-free graphs~\cite{JPP18}, respectively, for every $s\geq 0$.
In contrast to the case for {\sc Odd Cycle Transversal} and {\sc Connected Odd Cycle Transversal}, it is not known whether there is an integer~$r$ for which any of the problems {\sc Vertex Cover}, {\sc Feedback Vertex Set} or their connected variants is \NP-complete on $P_r$-free graphs.
Determining whether such an~$r$ exists is an interesting open problem. 

We note that a similar complexity study has also been undertaken for the independent variants of the problems {\sc Feedback Vertex Set} and {\sc Odd Cycle Transversal}.\footnote{{\sc Independent Vertex Cover}
can be seen as {\sc $2$-Colouring}, with the additional restriction that one of the colours can be used at most~$k$ times. This problem is polynomial-time solvable.}
In particular, {\sc Independent Feedback Vertex Set} and {\sc Independent Odd Cycle Transversal} are polynomial-time solvable on $P_5$-free graphs~\cite{BDFJP19}, but their complexity status is unknown on $P_6$-free graphs.
It is not known whether there is an integer~$r$ such that {\sc Independent Feedback Vertex Set} or {\sc Independent Odd Cycle Transversal} is \NP-complete on $P_r$-free graphs.

We conclude that in order to make any further progress, we must better understand the structure of $P_r$-free graphs.
This topic has been well studied in recent years, see also for example~\cite{GJPS17,GOPSSS18}.
However, more research and new approaches will be needed.

\bibliographystyle{abbrv}
\bibliography{mybib}

\appendix

\section{The Proof of Theorem~\ref{t-cvc}}\label{a-cvc}

{\bf This appendix is for reviewing purposes only.}
We will adapt, in a straightforward way, the proof from~\cite{JPP18} for showing that {\sc Connected Vertex Cover} is polynomial-time solvable on $(sP_1+\nobreak P_5)$-free graphs for every $s\geq 1$.

We need the following definitions and lemmas. Let $G=(V,E)$ be a graph.
The \emph{contraction} of an edge~$uv\in E$ deletes the vertices~$u$ and~$v$ and replaces them by a new vertex made adjacent to precisely those vertices that were adjacent to~$u$ or~$v$ in~$G$ (without introducing self-loops or multiple edges).
Recall that a linear forest is the disjoint union of one or more paths.
The following lemma is a straightforward observation.

\begin{lemma}\label{l-contract}
Let~$H$ be a linear forest and let~$G$ be a connected $H$-free graph.
Then the graph obtained from~$G$ after contracting an edge is also connected and~$H$-free.
\end{lemma}

We need the following lemmas given in~\cite{JPP18}.

\begin{lemma}[\cite{JPP18}]\label{l-2}
Let $s\geq 0$ and let~$G$ be a connected $(sP_1+\nobreak P_5)$-free graph.
Then~$G$ has a connected dominating set~$D$ that is either a clique or has size at most $2s^2+s+3$.
Moreover, $D$ can be found in $O(n^{2s^2+s+3})$ time.
\end{lemma}

\begin{lemma}[\cite{JPP18}]\label{l-double}
Let~$J$ be an independent set in a connected graph~$G$ such that~$J$ has a vertex~$y$ that is adjacent to every vertex of $G-J$.
Let~$J'$ consist of those vertices of $J\setminus \{y\}$ that have two adjacent neighbours in $G-J$ (or equivalently, in~$G$).
Then a subset~$S$ of the vertex set of~$G$ is a connected vertex cover of~$G$ that contains~$J$ if and only if $S\setminus J'$ is a connected vertex cover of $G-J'$ that contains $J\setminus J'$.
\end{lemma}

\noindent
We also need an auxiliary problem defined in~\cite{JPP18}.
Let~$G$ be a connected graph, let $J\subseteq V_G$ be a subset of the vertex set of~$G$ and let~$y$ be a vertex of~$J$.
We call say that a triple $(G,J,y)$ is \emph{cover-complete} if it has the following three properties:
\begin{enumerate}[(a)]
\item $J$ is an independent set;
\item $y$ is adjacent to every vertex of $G-J$;
\item the neighbours of each vertex in $J\setminus \{y\}$ form an independent set in $G-J$.
\end{enumerate}
This leads to the following optimization problem:

\problemdef{{\sc Connected Vertex Cover Completion}}{a cover-complete triple $(G,J,y)$.}{find a smallest connected vertex cover~$S$ of~$G$ such that $J\subseteq S$.}

\noindent
We also need the following two lemmas.

\begin{lemma}[\cite{JPP18}]\label{l-vcvc}
Let $(G,\{y\},y)$ be a cover-complete triple, where~$G$ is an $(sP_1+\nobreak P_5)$-free graph for some $s\geq 0$.
Then it is possible to compute a smallest connected vertex cover of~$G$ that contains~$y$ in $O(n^{s+14})$ time.
\end{lemma}

\begin{lemma}[\cite{JPP18}]\label{t-ind2}
For every $s\geq 0$, {\sc Connected Vertex Cover Completion} can be solved in $O(n^{2s+19})$ time for cover-complete triples $(G,J,y)$, where~$G$ is an $(sP_1+\nobreak P_5)$-free graph.
\end{lemma}

\noindent
We are now ready to prove Theorem~\ref{t-cvc}, which we restate below.
The proof mimics the proof of~\cite{JPP18} and as mentioned at the start of this section, we include it only for reviewing purposes.

\medskip
\noindent
\faketheorem{Theorem~\ref{t-cvc} (restated).}
{\itshape For every $s\geq 0$, {\sc Connected Vertex Cover Extension} can be solved in polynomial time on $(sP_1+\nobreak P_5)$-free graphs.}

\begin{proof}
Let~$G$ be an $(sP_1+\nobreak P_5)$-free graph on~$n$ vertices for some $s\geq 0$ and let $W\subseteq V(G)$ be a subset of vertices of~$G$.
We may assume without loss of generality that~$G$ is connected.
By Lemma~\ref{l-2} we can first compute in $O(n^{2s^2+s+3})$ time a connected dominating set~$D$ that either has size at most $2s^2+s+3$ or is a clique.
We note that, if~$D$ is a clique, any vertex cover of~$G$ contains all but at most one vertex of~$D$.
This leads to a case analysis where we guess the subset $D^*\subseteq D\setminus W$ of vertices not in a smallest connected vertex cover of~$G$ that contains~$W$.
That is, we choose a set of at most one vertex if~$D$ is a clique and a set of at most $|D\setminus W|$ vertices otherwise, and eventually look at all such sets.
As $|D\setminus W|\leq |D|\leq 2s^2+s+3$ if~$D$ is not a clique, the number of guesses is $O(n^{2s^2+s+3})$.
For each guess of~$D^*$, we compute a smallest connected vertex cover~$S_{D^*}$ that contains all vertices of $(D\setminus D^*)\cup W$ and no vertex of~$D^*$.
Then, at the end, we return one that has minimum size overall.
In particular we note that, since~$D$ is a connected dominating set of~$G$, $D\cup W$ is also a connected dominating set of~$G$.

Let~$D^*$ be a guess.
Before we start our case analysis we first prove the following claim.

\clm{\label{clm:1-1}We may assume, at the expense of an~$O(n^{16s^3+4})$ factor in the running time, that $D\setminus D^*$ is connected.}

\medskip
\noindent
We prove Claim~\ref{clm:1-1} as follows.
Suppose $D\setminus D^*$ is not connected.
Recall that~$G[D]$ is either a complete graph or has size at most $2s^2+s+3$.
In the first case, $G[D\setminus D^*]$ is connected.
Hence, the second case applies so~$D$ has size at most $2s^2+s+3$.
Let $v\in D\setminus D^*$.
As~$G$ is $(sP_1+\nobreak P_5)$-free, $G$ is also $P_{5+2s}$-free.
Hence, for each $u\in D\setminus (D^*\cup \{v\})$, every connected vertex cover of~$G$ contains a path of at most $5+2s-1$ vertices that connects~$u$ to~$v$.
We will guess all these $u-v$-paths (using only vertices from $G-D^*$) and add their vertices to~$D$.
As the number of paths is at most $2s^2+s+2$, this branching adds an $O(n^{(5+2s-3)(2s^2+s+2)})=O(n^{16s^3+4})$ factor to our running time and increases our set~$D$ by at most~$24s^3$ extra vertices.
We have proven Claim~\ref{clm:1-1}.

\medskip
\noindent
We distinguish two cases.

\thmcase{\label{case:1-1}$D^*=\emptyset$.}
We compute a minimum vertex cover~$S'$ of $G-(D\cup W)$ in polynomial time by Theorem~\ref{sp1p6-vc}.
To be more precise, this takes~$O(n^{s+14})$ time by using the same arguments as in the proof of Lemma~\ref{l-vcvc} (see~\cite{JPP18}).
Clearly $S'\cup D\cup W$ is a vertex cover of~$G$.
As~$D$ is a connected dominating set, $S'\cup D\cup W$ is even a connected vertex cover of~$G$.
Let $S_\emptyset=S'\cup D\cup W$.
As~$S'$ is a minimum vertex cover of $G-(D\cup W)$, $S_\emptyset$ is a smallest connected vertex cover of~$G$ that contains all vertices of~$D\cup W$.
We remember~$S_\emptyset$.
Note that~$S_\emptyset$ is found in~$O(n^{s+14})$ time.

\thmcase{\label{case:1-2}$1\leq |D^*|\leq |D|$\; (recall that $|D|\leq 2s^2+s+3$).}
Recall that we are looking for a smallest connected vertex cover of~$G$ that contains every vertex of $(D\setminus D^*)\cup W$, but does not contain any vertex of~$D^*$.
Hence~$D^*$ must be an independent set, disjoint from~$W$, and $G-D^*$ must be connected (if one of these conditions is false, then we stop considering the guess~$D^*$).
Moreover, a vertex cover that contains no vertex of~$D^*$ must contain all vertices of~$N_G(D^*)$.
Hence we can safely contract not only any edge between two vertices of $(D\setminus D^*)\cup W$, but also any edge between two vertices in~$N_G(D^*)$ or between a vertex of $(D\setminus D^*)\cup W$ and a vertex in~$N_G(D^*)$.
We perform edge contractions recursively and as long as possible while remembering all the edges that we contract.
This takes~$O(n)$ time.
Let~$G^*$ be the resulting graph.

Note that the set~$D^*$ still exists in~$G^*$, as we did not contract any edges with an endpoint in~$D^*$.
By Claim~\ref{clm:1-1}, the set $D\setminus D^*$ in~$G$ corresponds to exactly one vertex of~$G^*$.
We denote this vertex by~$y$.
The set~$W$ of~$G$ corresponds to an independent set of~$G^*$.
We denote this set by~$W^*$.
We observe the following equivalence, which is obtained after uncontracting all the contracted edges.

\clm{\label{clm:1-2}Every smallest connected vertex cover of~$G^*$ that contains $\{y\}\cup W^*$ and that does not contain any vertex of~$D^*$ corresponds to a smallest connected vertex cover of~$G$ that contains $(D\setminus D^*)\cup W$ and that does not contain any vertex of~$D^*$, and vice versa.}

\medskip
\noindent
As we obtained~$G^*$ in~$O(n)$ time, and we can also uncontract all contracted edges in~$O(n)$ time, Claim~\ref{clm:1-2} tells us that we may consider~$G^*$ instead of~$G$.
As~$G$ is connected and $(sP_1+\nobreak P_5)$-free, $G^*$ is also connected and $(sP_1+\nobreak P_5)$-free by Lemma~\ref{l-contract}.

We write $J^*=N_{G^*}(D^*)\cup W^*$ and note that~$y$ belongs to $N_{G^*}(D^*)\subseteq J^*$ as~$D$ is connected in~$G$.
We now consider the graph $G^*-D^*$.
As $G-D^*$ is connected, $G^*-D^*$ is connected.
By Claim~\ref{clm:1-2}, our new goal is to find a smallest connected vertex cover of $G^*-D^*$ that contains~$J^*$.
By our procedure, $J^*$ is an independent set of $G^*-D^*$.
As~$D$ dominates~$G$, we find that $D\setminus D^*$ dominates every vertex of $G-D^*$ that is not adjacent to a vertex of~$D^*$.
Hence the vertex~$y$, which corresponds to the set $D\setminus D^*$, is adjacent to every vertex of $(G^*-D^*)-J^*$ in the graph $G^*-D^*$.

Let $J\subseteq J^*$ consist of~$y$ and those vertices in~$J^*$ whose neighbourhood in $G^*-D^*$ is an independent set.
As~$y$ is adjacent to every vertex of $(G^*-D^*)-J^*$ in $G^*-D^*$, and we can remember the set $J^*\setminus J$, we can apply Lemma~\ref{l-double} and remove $J^*\setminus J$.
That is, it suffices to find a smallest connected vertex cover of the graph $G'=(G^*-D^*)-(J^*\setminus J)$ that contains~$J$.

As~$J^*$ is an independent set of $G^*-D^*$, we find that~$J$ is an independent set of~$G'$.
By definition, $y\in J$.
As~$y$ is adjacent to every vertex of $(G^*-D^*)-J^*$ in $G^*-D^*$, we find that~$y$ is adjacent to every vertex in $G'-J$.
By definition, the neighbours of each vertex in $J \setminus \{y\}$ form an independent set in $G'-J$.
Hence the triple $(G',J,y)$ is cover-complete.
This means that we can apply Lemma~\ref{t-ind2} to find in~$O(n^{2s+19})$ time a smallest connected vertex cover~$S'$ of~$G'$ that contains~$J$.

We translate~$S'$ in constant time into a smallest connected vertex cover~$S^*$ of $G^*-D^*$ that contains~$J^*$ by adding $J^*\setminus J$ to~$S'$.
We translate~$S^*$ in~$O(n)$ time into a smallest connected vertex cover~$S_{D^*}$ of~$G$ that contains $(D\setminus D^*)\cup W$ but no vertex of~$D^*$ by uncontracting any contracted edges.
It takes~$O(n^{2s+19})$ time to find the set~$S_{D^*}$.

\medskip
\noindent
As mentioned, at the end we pick a smallest set of the sets~$S_{D^*}$.
This set is then a smallest connected vertex cover of~$G$ that contains~$W$.
As there are $O(n^{2s^2+s+3}\cdot n^{16s^3+4})$ such sets, each of which is found in~$O(n^{2s+19})$ time, the total running time is $O(n^{21s^3 + 26})$.
The correctness of our algorithm follows immediately from the above case analysis and the description of the cases.\qed
\end{proof}

Note that the algorithm given in Theorem~\ref{t-cvc} not only solves the decision problem, but also finds a minimum connected vertex cover of a given $(sP_1+\nobreak P_5)$-free graph in polynomial time.

\end{document}